%% file: main.tex
\documentclass[lettersize,journal]{IEEEtran}

\usepackage[acronym]{glossaries}
\usepackage{amssymb}
\usepackage{amsmath}
\usepackage{amsthm}
\usepackage{graphicx}
\usepackage{setspace}
\usepackage[linesnumbered,ruled]{algorithm2e}
\usepackage{comment}
\loadglsentries{glossary}
\usepackage{xcolor}
\usepackage{lscape}
\usepackage{arydshln}
\usepackage{booktabs}

\newtheorem{lemma}{Lemma}
\newtheorem{proposition}{Proposition}

\begin{document}

\title{Pareto Frontier for the Performance-Complexity Trade-off in Beyond Diagonal\\Reconfigurable Intelligent Surfaces}

\author{Matteo~Nerini,~\IEEEmembership{Graduate Student Member,~IEEE},
        Bruno~Clerckx,~\IEEEmembership{Fellow,~IEEE}
        
\thanks{M. Nerini is with the Department of Electrical and Electronic Engineering, Imperial College London, London SW7 2AZ, U.K. (e-mail: m.nerini20@imperial.ac.uk).}
\thanks{B. Clerckx is with the Department of Electrical and Electronic Engineering, Imperial College London, London SW7 2AZ, U.K., and with Silicon Austria Labs (SAL), Graz A-8010, Austria (e-mail: b.clerckx@imperial.ac.uk).}}

\maketitle

\begin{abstract}
Reconfigurable intelligent surface (RIS) is an emerging technology allowing to control the propagation environment in wireless communications.
Recently, beyond diagonal RIS (BD-RIS) has been proposed to reach higher performance than conventional RIS, at the expense of higher circuit complexity.
Multiple BD-RIS architectures have been developed with the goal of reaching a favorable trade-off between performance and circuit complexity.
However, the fundamental limits of this trade-off are still unexplored.
In this paper, we fill this gap by deriving the expression of the Pareto frontier for the performance-complexity trade-off in BD-RIS.
Additionally, we characterize the optimal BD-RIS architectures reaching this Pareto frontier.
\end{abstract}

\glsresetall


\section{Introduction}

Reconfigurable intelligent surface (RIS) has recently gained a lot of popularity as a technology able to make the propagation environment smart and reconfigurable in wireless networks \cite{bas19,wu19,wu21}.
A RIS is composed of a large number of electrically tunable reflective elements that can be controlled to provide a passive beamforming gain.
Due to its ultra-low power consumption, low profile, and low cost, RIS is expected to efficiently improve future wireless communications.

In a conventional RIS architecture, also known as single-connected, each element is independently controlled by a tunable impedance component \cite{she20}.
This results in conventional RIS having a diagonal scattering matrix, also commonly known as a phase shift matrix.
To improve the capabilities of RIS, beyond diagonal RIS (BD-RIS) has been proposed as a generalization of conventional RIS, in which the scattering matrix is not limited to being diagonal \cite{li23-1}.
The key novelty introduced in BD-RIS is the presence of tunable impedance components interconnecting the RIS elements, adding further flexibility to the RIS at the expense of additional circuit complexity.
The single-connected RIS architecture has been first generalized in \cite{she20}.
By interconnecting some or all the RIS elements to each other, group- and fully-connected RIS architectures have been proposed, respectively \cite{she20}.
Group- and fully-connected RISs have been globally optimized in closed form assuming continuous reflection coefficients in \cite{ner22}, while they have been optimized using discrete reflection coefficients in \cite{ner21}.
Besides, BD-RIS have been modeled using graph theory in \cite{ner23}, where BD-RIS architectures have been described through graphs capturing the presence of tunable impedance components between the RIS elements.
Two low-complexity BD-RIS architectures have been proposed in \cite{ner23}, namely forest- and tree-connected RISs.

BD-RIS has been studied in several contexts showing significant performance gains over conventional RIS.
In \cite{li22-1}, a BD-RIS model has been developed unifying different BD-RIS working modes and different BD-RIS architectures.
In \cite{li22-2}, multi-sector BD-RIS has been introduced to efficiently enable full-space coverage.
Non-diagonal RIS \cite{li22} and dynamically group-connected RIS \cite{li22-3} have been proposed to outperform conventional RIS and group-connected RIS, respectively, thanks to their dynamic interconnections reconfigured on a per channel realization basis.
Additionally, BD-RIS has proved to enlarge the coverage and improve the sum rate in \gls{rsma} systems \cite{fan22,li23-2}, and to improve communication capacity and sensing precision in \gls{dfrc} systems \cite{wan23}.

When designing new BD-RIS architectures, the critical issue is the trade-off between performance and circuit complexity, given by the number of tunable impedance components in the BD-RIS architecture \cite{ner23}.
On the one hand, the single-connected RIS is the architecture with the lowest circuit complexity since there are no interconnections among the RIS elements.
Due to its limited flexibility, the single-connected RIS can only achieve a reduced performance.
On the other hand, the fully-connected RIS has the highest circuit complexity since each RIS element is connected to all others through a tunable impedance component, enabling the highest performance.
Several BD-RIS architectures have been proposed to trade performance and complexity.
However, the fundamental limits of this trade-off are still unexplored.
To fill this gap, we investigate how to optimally trade the achievable performance and the circuit complexity in BD-RIS architectures.

The contribution of this letter is twofold.
\textit{First}, we derive the Pareto frontier for the performance-complexity trade-off offered by BD-RIS in \gls{siso} systems.
\textit{Second}, we characterize the optimal BD-RIS architectures allowing to achieve this Pareto frontier.

\section{System Model}
\label{sec:system-model}

Consider a \gls{siso} communication system aided by an $N$-element RIS.
The $N$ elements of the RIS are connected to a $N$-port reconfigurable impedance network, with scattering matrix $\boldsymbol{\Theta}\in\mathbb{C}^{N\times N}$.
Defining $x\in\mathbb{C}$ as the transmitted signal and $y\in\mathbb{C}$ as the received signal, we have $y=hx+n$, where $h\in\mathbb{C}$ is the wireless channel and $n\in\mathbb{C}$ is the \gls{awgn} at the receiver.
Assuming that the direct link between the transmitter and the receiver is negligible compared to the RIS-aided link, the channel $h$ writes as $h=\mathbf{h}_{R}\boldsymbol{\Theta}\mathbf{h}_{T}$, where $\mathbf{h}_{R}\in\mathbb{C}^{1\times N}$ and $\mathbf{h}_{T}\in\mathbb{C}^{N\times 1}$ refer to the channels from the RIS to the receiver and from the transmitter to the RIS, respectively \cite{she20}\footnote{Since $\mathbf{h}_{R}\boldsymbol{\Theta}\mathbf{h}_{T}$ can always be co-phased with the direct link, our conclusions are not impacted by the direct link.
In the case of a non-negligible direct link, the performance would merely be scaled up, depending on its strength.
Thus, we neglect the direct link to gain fundamental insights not depending on its strength.}.
We assume independent and identically distributed (i.i.d.) Rayleigh channels to obtain fundamental insights, having unit channel gains with no loss of generality, i.e., $\mathbf{h}_{R}\sim\mathcal{CN}\left(\boldsymbol{0},\mathbf{I}\right)$ and $\mathbf{h}_{T}\sim\mathcal{CN}\left(\boldsymbol{0},\mathbf{I}\right)$.

When reconfiguring a RIS, the scattering matrix $\boldsymbol{\Theta}$ is typically optimized to maximize the performance given by the received signal power
\begin{equation}
P_R=P_T\left|\mathbf{h}_{R}\boldsymbol{\Theta}\mathbf{h}_{T}\right|^2,
\end{equation}
where $P_T=\mathrm{E}[\left|x\right|^2]$ is the transmitted signal power.
Considering passive RISs with lossless and reciprocal impedance networks, the matrix $\boldsymbol{\Theta}$ is in general subject to the constraints $\boldsymbol{\Theta}^{H}\boldsymbol{\Theta}=\boldsymbol{\mathrm{I}}$ and $\boldsymbol{\Theta}=\boldsymbol{\Theta}^{T}$ \cite{poz11}.
Furthermore, additional constraints on $\boldsymbol{\Theta}$, limiting the received signal power, are present depending on the BD-RIS architecture \cite{she20,ner23}.

\section{Problem Formulation}

Conventional RIS, also known as single-connected RIS, is the least complex architecture achieving the lowest performance, given by
\begin{equation}
\bar{P}_R^{\mathrm{Single}}=P_T\left(\sum_{n=1}^N\left\vert\left[\mathbf{h}_{R}\right]_{n}\left[\mathbf{h}_{T}\right]_{n}\right\vert\right)^2,
\end{equation}
since it includes only $N$ tunable impedance components \cite{she20}.
In contrast, tree-connected RIS is proved to be the least complex architecture achieving the performance upper bound
\begin{equation}
\bar{P}_{R}^{\mathrm{Tree}}=P_T\left\Vert\mathbf{h}_{R}\right\Vert^2\left\Vert\mathbf{h}_{T}\right\Vert^2,
\end{equation}
with $2N-1$ tunable impedance components \cite{ner23}.
In this letter, our goal is to determine the maximum performance achievable by BD-RIS architectures with circuit complexity $C\in[N,2N-1]$, representing the number of tunable components\footnote{In our analysis, we preclude BD-RISs with dynamic interconnections, since they require switches and hence additional circuit complexity.}.
In other words, we want to characterize the Pareto frontier of the performance-complexity trade-off enabled by BD-RIS.
Furthermore, we are interested in which BD-RIS architectures allow us to reach such a frontier, denoted as ``optimal'' BD-RIS architectures in the following.

We begin by characterizing the maximum received signal power achievable by a given BD-RIS architecture.
To this end, we consider the modeling of BD-RIS based on graph theory developed in \cite{ner23}.
According to \cite{ner23}, each BD-RIS architecture can be described through a graph $\mathcal{G}$ capturing the presence of tunable impedance components between its RIS elements.
We denote as $G$ the number of connected components of such a graph $\mathcal{G}$, where a connected component of a graph is defined as a connected subgraph that is not part of any larger connected subgraph \cite{bon76}. Besides, $N_g\geq1$ is the number of RIS elements included in the $g$th component, with $\sum_{g=1}^GN_{g}=N$.
In agreement with previous work on BD-RIS \cite{she20}-\cite{li22-3}, we refer to the connected components of $\mathcal{G}$ as the ``groups'' of the corresponding BD-RIS architecture.
According to \cite{ner23}, the maximum received signal power obtained by the BD-RIS associated with $\mathcal{G}$ is given by
\begin{equation}
\bar{P}_R=P_T\left(\sum_{g=1}^G\left\Vert\mathbf{h}_{R,g}\right\Vert\left\Vert\mathbf{h}_{T,g}\right\Vert\right)^2,\label{eq:PR}
\end{equation}
where $\mathbf{h}_{R,g}\in\mathbb{C}^{1\times N_{g}}$ and $\mathbf{h}_{T,g}\in\mathbb{C}^{N_{g}\times1}$ contain the $N_{g}$ elements of $\mathbf{h}_{R}$ and $\mathbf{h}_{T}$ corresponding to the $N_{g}$ RIS elements included into the $g$th group, respectively.
In the case of i.i.d. fading channels, we can assume that each group includes adjacent RIS elements with no loss of generality, such that $\mathbf{h}_{R}=[\mathbf{h}_{R,1},\ldots,\mathbf{h}_{R,G}]$ and $\mathbf{h}_{T}=[\mathbf{h}_{T,1}^T,\ldots,\mathbf{h}_{T,G}^T]^{T}$.
Thus, the maximum received signal power $\bar{P}_R$ achievable by a given BD-RIS solely depends on $G$ and the group sizes $N_1,\ldots,N_G$.

To express the maximum received signal power $\bar{P}_R$ achievable with a circuit complexity $C$ as a function of $C$, we introduce the following three results.
First, we characterize the optimal BD-RIS architectures through the following lemma.
\begin{lemma}
All the optimal BD-RIS architectures have a corresponding graph being acyclic, also known as a forest.
\label{lem:forest1}
\end{lemma}
\begin{proof}
Please refer to Appendix~A.
\end{proof}
In other words, a BD-RIS architecture can be optimal only if its graph does not contain any cycle, i.e., a finite sequence of distinct edges joining a sequence of vertices, where only the first and last vertices are equal \cite{bon76}.
Second, we use the following result from graph theory \cite{bon76}.
\begin{lemma}
If a graph $\mathcal{G}$ is a forest, then it has $G=N-L$ connected components, where $N$ is the number of vertices and $L$ is the number of edges.
\label{lem:forest2}
\end{lemma}
\begin{proof}
Please refer to Appendix~B.
\end{proof}
Third, by using Lemma~\ref{lem:forest1} and Lemma~\ref{lem:forest2}, we can derive the following proposition.
\begin{proposition}
An optimal BD-RIS architecture with $N$ elements and circuit complexity $C$, with $C\in[N,2N-1]$, has a corresponding graph with $G=2N-C$ connected components.
\label{pro:groups}
\end{proposition}
\begin{proof}
Please refer to Appendix~C.
\end{proof}
According to Proposition~\ref{pro:groups}, given a circuit complexity $C$, the number of groups $G$ in the corresponding optimal BD-RIS is fixed.
Thus, our problem is to find the group sizes $N_1,\ldots,N_G$ of the BD-RIS architecture that maximize the performance $\mathrm{E}\left[\bar{P}_R\right]$, with fixed $G$.
The corresponding optimization problem is given by
\begin{align}
\underset{N_1,\ldots,N_G}{\mathsf{\mathrm{max}}}\;\;
&\mathrm{E}\left[\bar{P}_R\right]\label{eq:p1-obj}\\
\mathsf{\mathrm{s.t.}}\;\;\;
&N_g\geq1,\:\forall g,\:\:\sum_{g=1}^GN_{g}=N,\label{eq:p1-c}
\end{align}
where $G=2N-C$ is fixed depending on the complexity $C$.

\section{Pareto Frontier}

We now solve problem \eqref{eq:p1-obj}-\eqref{eq:p1-c} by rewriting and simplifying the objective \eqref{eq:p1-obj}, then we use the obtained optimal group sizes $N_1,\ldots,N_G$ to derive the desired Pareto frontier.
To solve problem \eqref{eq:p1-obj}-\eqref{eq:p1-c}, we assume $P_T=1$ with no loss of generality and we write the average received signal power \eqref{eq:PR} as
\begin{multline}
\mathrm{E}\left[\bar{P}_R\right]=\sum_{g=1}^G\mathrm{E}\left[\left\Vert\mathbf{h}_{R,g}\right\Vert^2\right]^2\\
+\sum_{g_1\neq g_2}\mathrm{E}\left[\left\Vert\mathbf{h}_{R,g_1}\right\Vert\right]^2\mathrm{E}\left[\left\Vert\mathbf{h}_{R,g_2}\right\Vert\right]^2,
\end{multline}
where we exploited the i.i.d. channels assumption and the fact that $\mathbf{h}_{R}$ and $\mathbf{h}_{T}$ are identically distributed.
Using the moments of the chi distribution with $2N_g$ degrees of freedom, we have that $\mathrm{E}[\Vert\mathbf{h}_{R,g}\Vert]=\Gamma(N_g+1/2)/\Gamma(N_g)$ and $\mathrm{E}[\Vert\mathbf{h}_{R,g}\Vert^2]=N_g$, $\forall g$, where $\Gamma(\cdot)$ is the gamma function.
Thus, we can write
\begin{multline}
\mathrm{E}\left[\bar{P}_R\right]=\sum_{g=1}^GN_g^2\\
+\sum_{g_1\neq g_2}\left(\frac{\Gamma\left(N_{g_1}+1/2\right)}{\Gamma\left(N_{g_1}\right)}\right)^2\left(\frac{\Gamma\left(N_{g_2}+1/2\right)}{\Gamma\left(N_{g_2}\right)}\right)^2.\label{eq:EPR1}
\end{multline}


The expression of $\mathrm{E}\left[\bar{P}_R\right]$ in \eqref{eq:EPR1} can be now simplified by using the relationship
\begin{equation}
\left(\frac{\Gamma\left(M+1/2\right)}{\Gamma\left(M\right)}\right)^2=M-\frac{1}{4}+\frac{1}{32M}+\mathcal{O}\left(\left(\frac{1}{M}\right)^2\right),
\end{equation}
given by the Laurent series expansion at $M=\infty$ \cite{ahl53}.
Remarkably, the function $(\Gamma(M+1/2)/\Gamma(M))^2$ is well approximated by $M-1/4+1/(32M)$ for any positive integer $M$, despite the series being computed at $M=\infty$.
Thus, we can approximate \eqref{eq:EPR1} as
\begin{multline}
\mathrm{E}\left[\bar{P}_R\right]=\sum_{g=1}^GN_g^2\\
+\sum_{g_1\neq g_2}\left(N_{g_1}-\frac{1}{4}+\frac{1}{32N_{g_1}}\right)\left(N_{g_2}-\frac{1}{4}+\frac{1}{32N_{g_2}}\right).
\end{multline}
By developing the product and reorganizing the terms, we get
\begin{multline}
\mathrm{E}\left[\bar{P}_R\right]=\sum_{g=1}^GN_g^2+\sum_{g_1\neq g_2}\left(N_{g_1}N_{g_2}+\frac{1}{16}-\frac{N_{g_1}}{4}-\frac{N_{g_2}}{4}\right.\\
+\left.\frac{N_{g_1}}{32N_{g_2}}+\frac{N_{g_2}}{32N_{g_1}}-\frac{1}{128N_{g_1}}-\frac{1}{128N_{g_2}}+\frac{1}{1024N_{g_1}N_{g_2}}\right).\label{eq:EPR2}
\end{multline}
Recalling that $N_g\geq1$, $\forall g$, the term $1/(1024N_{g_1}N_{g_2})$ in \eqref{eq:EPR2} is negligible since it is at least 1024 times smaller than the term $N_{g_1}N_{g_2}$.
Thus, omitting $1/(1024N_{g_1}N_{g_2})$ and completing the computations, we obtain
\begin{multline}
\mathrm{E}\left[\bar{P}_R\right]=N^2+\frac{G\left(G-1\right)}{16}-\frac{N\left(G-1\right)}{2}\\
+\sum_{g_1\neq g_2}\left(\frac{N_{g_1}}{32N_{g_2}}+\frac{N_{g_2}}{32N_{g_1}}-\frac{1}{128N_{g_1}}-\frac{1}{128N_{g_2}}\right).
\end{multline}
Observing that
\begin{gather}
\sum_{g_1\neq g_2}\frac{N_{g_1}}{32N_{g_2}}=\sum_{g_1\neq g_2}\frac{N_{g_2}}{32N_{g_1}}=\sum_{g=1}^G\frac{N-N_{g}}{32N_{g}},\\
\sum_{g_1\neq g_2}\frac{1}{128N_{g_1}}=\sum_{g_1\neq g_2}\frac{1}{128N_{g_2}}=\sum_{g=1}^G\frac{G-1}{128N_{g}},
\end{gather}
we can eventually rewrite $\mathrm{E}\left[\bar{P}_R\right]$ as
\begin{multline}
\mathrm{E}\left[\bar{P}_R\right]=N^2+\frac{G-1}{16}\left(G-8N\right)\\
-\frac{G}{16}+\frac{4N-G+1}{64}\sum_{g=1}^G\frac{1}{N_{g}}.\label{eq:EPR3}
\end{multline}

\begin{figure*}[t]
\centering
\includegraphics[width=0.40\textwidth]{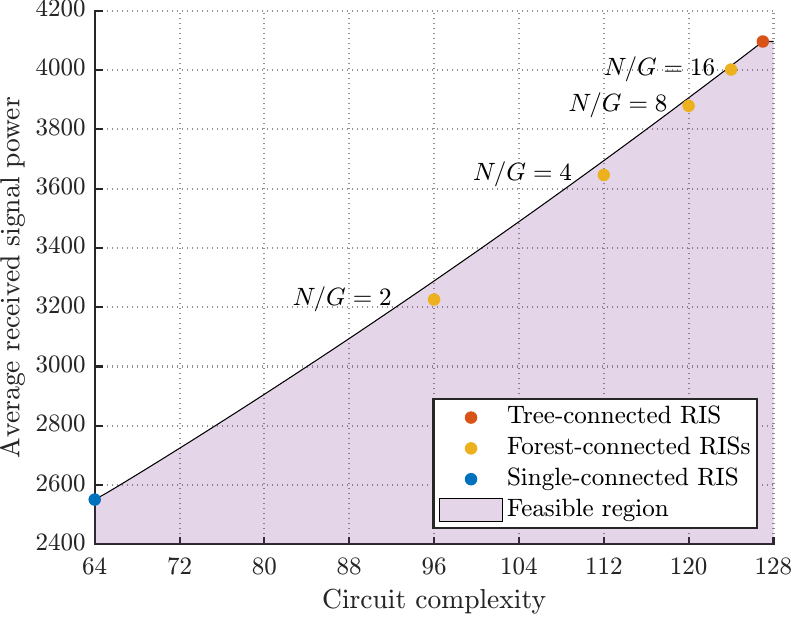}
\includegraphics[width=0.40\textwidth]{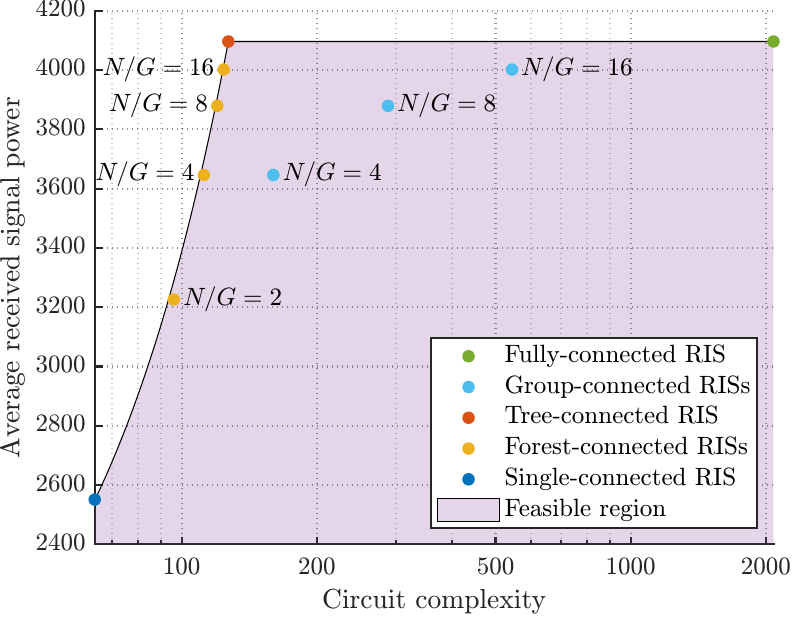}
\caption{Pareto frontier for the performance-complexity trade-off achieved by BD-RISs, with $N=64$.}
\label{fig:pareto}
\end{figure*}

We notice from \eqref{eq:EPR3} that maximizing $\mathrm{E}\left[\bar{P}_R\right]$ is equivalent to maximize $\sum_{g=1}^G1/N_{g}$. Thus, problem \eqref{eq:p1-obj}-\eqref{eq:p1-c} can be equivalently expressed as
\begin{align}
\underset{N_1,\ldots,N_G}{\mathsf{\mathrm{max}}}\;\;
&\sum_{g=1}^G\frac{1}{N_{g}}\label{eq:p2-obj}\\
\mathsf{\mathrm{s.t.}}\;\;\;
&N_g\geq1,\:\forall g,\:\:\sum_{g=1}^GN_{g}=N,\label{eq:p2-c}
\end{align}
which is solved in the following proposition.
\begin{proposition}
The solution to problem \eqref{eq:p2-obj}-\eqref{eq:p2-c} is given by
\begin{gather}
N_1=N_2=\ldots=N_{G-1}=1,\label{eq:opt-NG1}\\
N_G=N-G+1,\label{eq:opt-NGG}
\end{gather}
up to a permutation of the group sizes.
\label{pro:solution}
\end{proposition}
\begin{proof}
Please refer to Appendix~D.
\end{proof}
Given the optimal group sizes $N_1,\ldots,N_{G}$ provided by Proposition~\ref{pro:solution}, we can derive in closed form the expression of the desired Pareto frontier.
Specifically, plugging \eqref{eq:opt-NG1} and \eqref{eq:opt-NGG} into \eqref{eq:EPR1}, we obtain
\begin{multline}
\mathrm{E}\left[\bar{P}_R\right]=G-1+\left(N-G+1\right)^2+\left(G-1\right)\left(G-2\right)\Gamma\left(3/2\right)^4\\
+2\left(G-1\right)\left(\frac{\Gamma\left(N-G+3/2\right)\Gamma\left(3/2\right)}{\Gamma\left(N-G+1\right)}\right)^2,
\end{multline}
giving the maximum performance achievable with a BD-RIS architecture having $G$ groups.
Finally, recalling that $G=2N-C$, the expression of the maximum performance achievable with a circuit complexity $C\in[N,2N-1]$ is given by
\begin{multline}
\mathrm{E}\left[\bar{P}_R\right]=\left(C-N\right)^2+C\\
+\left(2N-C-1\right)\left(2N-C-2\right)\Gamma\left(3/2\right)^4\\
+2\left(2N-C-1\right)\left(\frac{\Gamma\left(C-N+3/2\right)\Gamma\left(3/2\right)}{\Gamma\left(C-N+1\right)}\right)^2,\label{eq:pareto}
\end{multline}
representing the Pareto frontier of the performance-complexity trade-off offered by BD-RISs.

\section{Numerical Results}

In Fig.~\ref{fig:pareto}, we report the Pareto frontier given by \eqref{eq:pareto}, delimiting the region of feasible BD-RIS architectures\footnote{Note that the average received signal power has no units since is computed with unit transmit power and unit channel gains with no loss of generality.}.
Specifically, the feasible region is delimited by \eqref{eq:pareto} when $C\in[N,2N-1]$ since \eqref{eq:pareto} gives the maximum performance achievable with complexity $C$, and by the horizontal line $\mathrm{E}[\bar{P}_R]=N^2$ when $C>2N-1$ since $N^2$ is the performance upper bound \cite{she20}.
We fix $N=64$ in Fig.~\ref{fig:pareto} to obtain a fair comparison in terms of the space occupied by the RIS.

The derived frontier is compared with the performance-complexity trade-off achieved by the BD-RIS architectures recently proposed in \cite{she20,ner23}.
Since each BD-RIS is characterized by its circuit complexity $C$ and
average received signal power $\mathrm{E}[\bar{P}_R]$, each BD-RIS is represented as a point in Fig.~\ref{fig:pareto} with coordinates $(C,\mathrm{E}[\bar{P}_R])$.
More precisely, we report group- and forest-connected RISs, both achieving a performance
\begin{equation}
\mathrm{E}\left[\bar{P}_R^{\mathrm{Group}}\right]=\frac{N^2}{G}+G\left(G-1\right)\left(\frac{\Gamma\left(N/G+1/2\right)}{\Gamma\left(N/G\right)}\right)^4,\label{eq:EPR-group}
\end{equation}
and with complexity $C^{\mathrm{Group}}=N(N/G+1)/2$ and $C^{\mathrm{Forest}}=2N-G$, respectively, depending on the group size $N/G$ \cite{she20,ner23}.
Note that \eqref{eq:EPR-group} is obtained by setting $N_g=N/G$, $\forall g$, in \eqref{eq:EPR1}.
The four forest-connected RISs in Fig.~\ref{fig:pareto} have group sizes 2, 4, 8, and 16, while the three group-connected RISs have group sizes 4, 8, and 16 since forest- and group-connected RISs with group sizes 2 are equivalent \cite{ner23}.
Besides, we report fully- and tree-connected RISs, both achieving 
\begin{equation}
\mathrm{E}\left[\bar{P}_R^{\mathrm{Fully}}\right]=N^2,\label{eq:EPR-fully}
\end{equation}
and with complexity $C^{\mathrm{Fully}}=N(N+1)/2$ and $C^{\mathrm{Tree}}=2N-1$, respectively \cite{she20,ner23}, where \eqref{eq:EPR-fully} is derived by setting $G=1$ in \eqref{eq:EPR-group}.
Finally, we report the single-connected RIS architecture, achieving a performance given by
\begin{equation}
\mathrm{E}\left[\bar{P}_R^{\mathrm{Single}}\right]=N+N\left(N-1\right)\Gamma\left(3/2\right)^4,\label{eq:EPR-single}
\end{equation}
and with complexity $C^{\mathrm{Single}}=N$ \cite{she20}, where \eqref{eq:EPR-single} is derived by setting $G=N$ in \eqref{eq:EPR-group}.

We make the following remarks.
\textit{First}, on the one hand, the single-connected RIS is the least complex architecture, achieving the lowest performance due to its limited architecture.
On the other hand, the tree-connected RIS allows us to reach the performance upper bound, with the lowest possible complexity.
\textit{Second}, forest-connected RISs approach the Pareto frontier, but they are slightly suboptimal.
This is because forest-connected RISs have equally sized groups, i.e., they all have group size $N/G$.
However, the optimal group sizes are not all equal, as given by \eqref{eq:opt-NG1}-\eqref{eq:opt-NGG}.
\textit{Third}, the fully-connected (resp. group-connected) RIS achieves the same performance as the tree-connected (resp. forest-connected) RIS, but with higher circuit complexity.
Thus, in \gls{siso} systems, fully- and group-connected RISs are highly suboptimal.
Note that the exact shape of the Pareto frontier depends on the channel distribution.
Specifically, with Rician or correlated channels, the gain of the tree-connected over the single-connected RIS decreases, and less complex BD-RIS architectures are expected to approach the performance upper bound \cite{she20,ner23}.

\section{Conclusion}

We derive the Pareto frontier for the performance-complexity trade-off in BD-RISs.
This frontier provides the BD-RIS architectures that can be optimally used to bridge between the single-connected RIS and the tree-connected RIS.
The presented fundamental results are expected to drive the development of novel BD-RIS architectures and prototypes.
Remarkably, a multi-dimensional Pareto frontier could be derived accounting for multiple parameters in addition to the circuit complexity, such as the number of RIS elements, the channel estimation overhead, and the optimization complexity, thus representing a future research direction.

\section*{Appendix}

\subsection{Proof of Lemma~\ref{lem:forest1}}

To prove Lemma~\ref{lem:forest1}, we show that any BD-RIS whose corresponding graph is not a forest, is not optimal.
Consider a BD-RIS with circuit complexity $C$ with a corresponding graph $\mathcal{G}$ that is not a forest, i.e., it has at least one cycle \cite{bon76}.
According to \eqref{eq:PR}, the received signal power achievable by a BD-RIS solely depends on its number of groups $G$, and the group sizes $N_1,\ldots,N_G$.
Thus, by removing one edge from a cycle in $\mathcal{G}$, the resulting BD-RIS has complexity $C-1$ and achieves the same performance as the original one since its graph still has $G$ connected components with sizes $N_1,\ldots,N_G$.
Since this resulting BD-RIS achieves the same performance as the original one with reduced complexity, the original BD-RIS is not optimal, and Lemma~\ref{lem:forest1} is proved.

\subsection{Proof of Lemma~\ref{lem:forest2}}

Consider a forest $\mathcal{G}$ with $G$ connected components, where the $g$th component has $N_g$ vertices, with $\sum_{g=1}^GN_g=N$.
Since $\mathcal{G}$ is a forest, each component is a tree, i.e., is a connected forest, and the $g$th component includes $N_g-1$ edges \cite[Theorem 2.2]{bon76}.
Thus, the number of edges in $\mathcal{G}$ is
\begin{equation}
L=\sum_{g=1}^G\left(N_g-1\right)=\sum_{g=1}^GN_g-G=N-G,
\end{equation}
proving that $G=N-L$.

\subsection{Proof of Proposition~\ref{pro:groups}}

According to \cite{ner23}, a BD-RIS with $N$ elements and $C$ tunable impedance components has a corresponding graph with $L=C-N$ edges since $N$ tunable impedance components are used to connect each RIS element to ground.
Besides, we know from Lemma~\ref{lem:forest1} and Lemma~\ref{lem:forest2} that the graph of an optimal BD-RIS with $N$ elements and complexity $C$ is a forest, having $G=N-L$ connected components.
By plugging $L=C-N$ into $G=N-L$, we obtain $G=2N-C$.

\subsection{Proof of Proposition~\ref{pro:solution}}

The proof is conducted by induction on the number of groups $G$.
As the base case, we consider $G=2$, where problem \eqref{eq:p2-obj}-\eqref{eq:p2-c} boils down to
\begin{align}
\underset{N_1,N_2}{\mathsf{\mathrm{min}}}\;\;
&N_1N_2\label{eq:obj}\\
\mathsf{\mathrm{s.t.}}\;\;\;
&N_1\geq1,\:N_2\geq1,\:N_1+N_2=N.
\end{align}
The solution to this problem is clearly given by $N_1=1$ and $N_2=N-1$, or vice versa.
The proposition is consequently verified for the case $G=2$.

As the induction step, we prove that if the proposition is valid for $G-1$ groups, it also holds for $G$ groups.
To this end, we rewrite problem \eqref{eq:p2-obj}-\eqref{eq:p2-c} as
\begin{align}
\underset{N_1,\ldots,N_G}{\mathsf{\mathrm{max}}}\;\;
&\sum_{g=1}^{G-1}\frac{1}{N_{g}}+\frac{1}{N_{G}}\label{eq:p3-obj}\\
\mathsf{\mathrm{s.t.}}\;\;\;
&N_g\geq1,\:\forall g,\:\:\sum_{g=1}^{G-1}N_{g}=N-N_G.\label{eq:p3-c}
\end{align}
By the induction hypothesis, we have
\begin{gather}
N_1=N_2=\ldots=N_{G-2}=1,\label{eq:NG12}\\
N_{G-1}=N-N_G-G+2.\label{eq:NGG2}
\end{gather}
Using \eqref{eq:NG12} and \eqref{eq:NGG2}, problem \eqref{eq:p3-obj}-\eqref{eq:p3-c} can be simplified as
\begin{align}
\underset{N_{G-1},N_G}{\mathsf{\mathrm{min}}}\;\;
&N_{G-1}N_G\\
\mathsf{\mathrm{s.t.}}\;\;\;
&N_{G-1}\geq1,\:N_G\geq1,\\
&N_{G-1}+N_G=N-G+2,
\end{align}
where the only unknown are $N_{G-1}$ and $N_G$.
By solving this problem as done for the base case, we obtain $N_{G-1}=1$ and $N_G=N-G+1$, proving the induction step.

\bibliographystyle{IEEEtran}
\bibliography{IEEEabrv,main}
\end{document}

%% file: glossary.tex
\newacronym{awgn}{AWGN}{additive white Gaussian noise}
\newacronym{siso}{SISO}{single-input single-output}
\newacronym{star-ris}{STAR-RIS}{simultaneously transmitting and reflecting RIS}
\newacronym{ios}{IOS}{intelligent omni-surface}
\newacronym{rsma}{RSMA}{rate splitting multiple access}
\newacronym{dfrc}{DFRC}{dual-function radar-communication}